\documentclass[journal]{IEEEtran}
\usepackage{epsfig}
\usepackage{amsmath}
\usepackage{amsfonts}
\usepackage{bbm}
\usepackage{cite}
\usepackage{graphicx}
\usepackage{multirow}
\usepackage{longtable}
\usepackage{txfonts}
\usepackage{amsfonts} 
\newtheorem{theorem}{Theorem}

\newtheorem{lemma}{Lemma}
\newtheorem{proposition}{Proposition}

\newtheorem{remark}{Remark}
\newtheorem{example}{Example}

\newcommand{\tabincell}[2]{\begin{tabular}{@{}#1@{}}#2\end{tabular}}

\begin{document}

\title{Period Distribution of Inversive Pseudorandom Number Generators Over Galois Rings}
\author{ Bo~Zhou~\IEEEmembership{}
        and~Qiankun~Song\IEEEmembership{}
\thanks{~This work was partially supported by the National Natural Science Foundation
of China under Grant 60974132, the Natural Science Foundation Project of CQ CSTC2011BA6026 and the Scientific \& Technological Research Projects of CQ KJ110424.}
\thanks{~B. Zhou is with College of Information Science \& Engineering, Chongqing Jiaotong University, Chongqing 400074, P.R. China. (e-mail:
zhoubocncq@163.com) }
\thanks{~Q.K. Song is with Department of Mathematics, Chongqing Jiaotong University,
Chongqing 400074, P.R. China. (e-mail:
qiankunsong@163.com)}
}

\markboth{SUBMITTED TO IEEE TRANSACTIONS ON INFORMATION THEORY}{Period Distribution of Inversive Pseudorandom Number Generators Over Galois Ring ${\rm Z}_{p^{e}}$}%

\maketitle
\date{}

\begin{abstract}
In 2009, Sol\'{e} and Zinoviev (\emph{Eur. J. Combin.}, vol. 30, no. 2, pp. 458-467, 2009) proposed an open problem of arithmetic interest to study the period of the inversive pseudorandom number generators (IPRNGs) and to give conditions bearing on $a, b$ to achieve maximal period, we focus on resolving this open problem. In this paper, the period distribution of the IPRNGs over the Galois ring $({\rm Z}_{p^{e}},+,\times)$ is considered, where $p>3$ is a prime and $e\geq 2$ is an integer. The IPRNGs are transformed to $2$-dimensional linear feedback shift registers (LFSRs) so that the analysis of the period distribution of the IPRNGs is transformed to the analysis of the period distribution of the LFSRs. Then, by employing some analytical approaches, the full information on the period distribution of the IPRNGs is obtained, which is to make exact statistics about the period of the IPRNGs then count the number of IPRNGs of a specific period when $a$, $b$ and $x_{0}$ traverse all elements in ${\rm Z}_{p^{e}}$. The analysis process also indicates how to choose the parameters and the initial values such that the IPRNGs fit specific periods.
\end{abstract}
\begin{IEEEkeywords}
Inversive pseudorandom number generator (IPRNG), linear feedback shift register (LFSR), period distribution, Galois ring.
\end{IEEEkeywords}

\IEEEpeerreviewmaketitle

\section{Introduction}
A pseudorandom number generator (PRNG) is a deterministic algorithm that produces a long sequence of numbers that appear random and indistinguishable from a stream of random numbers, which is widely employed in engineering applications, e.g., generation of cryptographic keys and random initialization of certain variables in cryptographic protocols \cite{s17}. PRNGs are implemented on finite-state machines, thus, the sequences generated by them are ultimately periodic. In cryptographic applications of PRNGs, a long period is often required. In this case, the full information on the period distribution of the PRNGs plays an important role. If the full information on the period distribution of PRNGs is known, one will be able to choose the suitable parameters and initial values such that the PRNGs fit specific periods.

In \cite{c4,c5,c6,c3}, the detailed period distribution of several linear map based PRNGs, such as the Arnold cat map \cite{a11} and the Chebychev map \cite{k12}, have been studied. In \cite{e8}, a nonlinear map based PRNG called IPRNG was proposed, which is shown as follows:
$$
x_{n+1}=\left\{\begin{array}{cccc}
a x_{n}^{-1}+b {\rm mod}p,&x_{n}\neq0\\
b,&x_{n}=0
\end{array}\right.,
$$
for all $n\geq0$, where $a,b\in{\rm GF}(p)$ and its initial value is $x_{0}\in{\rm GF}(p)$.

Soon afterwards, the study on the properties of IPRNGs has become a hot topic. In \cite{g16,n17,n10,n14}, the distribution properties of the IPRNGs were studied. In \cite{g15}, the complexity profile of the IPRNGs was considered. In \cite{c9}, the period of the IPRNGs was  investigated, the considered state space was a Galois field, but the authors did not provide the full information on the period distribution of IPRNGs. Here, we will further consider the full information on the period distribution of IPRNGs over $({\rm Z}_{p^{e}},+,\times)$. However, the structure of $({\rm Z}_{p^{e}},+,\times)$ is more complicated than which of Galois fields, that is, $({\rm Z}_{p^{e}},+,\times)$ contains many zero divisors but the Galois field does not.

In 2009, Sol\'{e} and Zinoviev \cite{s7} provided a novel construction of IPRNGs as follows:
$$
\phi(p^{k}x)=\left\{\begin{array}{cccc}
p^{k}a x^{-1}+b &x\in{\rm R}^{\times}\\
b&x=0
\end{array}\right.,
$$
where ${\rm R}$ is a Galois ring, ${\rm R}^{\times}$ is the group of units of ${\rm R}$, $\phi$ is the map from ${\rm R}$ to itself,  $a,b\in {\rm R}^{\times}$ and $x_{0}\in{\rm R}$. The discrepancy estimates of the IPRNGs both for the full period and for certain special parts of the period was considered. In order to generalize these estimates to arbitrary parts of the period, the authors proposed an open problem of arithmetic interest to study the period of the inversive pseudorandom number generators and to give conditions bearing on $a, b$ to achieve maximal period.

Motivated by the above discussions, we focus on analyzing the full information on the period distribution of the inversive pseudorandom number generators (IPRNGs) over the Galois ring $({\rm Z}_{p^{e}},+,\times)$, where $p>3$ is a prime and $e\geq 2$ is an integer. The IPRNGs considered in this paper are transformed to $2$-dimensional LFSRs so that the analysis of the period distribution of the IPRNGs is transformed to the analysis of the period distribution of the LFSRs. Then, the full information on the period distribution of IPRNGs is obtained by some analytical approaches, i.e., analyzing the general terms of the LFSRs and the order of the roots of the characteristic polynomial of the LFSRs. The analysis process also indicates how to choose the parameters and the initial values such that the IPRNGs fit specific periods. It is noteworthy that the analysis of the order of the roots of the polynomials is also useful in the analysis of the period of the polynomials which is an interesting problem in the analysis of sequences over Galois rings \cite{f18,f19,f20,f22}.
\section{Preliminaries}
In this section, some concepts and notations on Galois rings and IPRNGs employed in this paper are introduced. For more detailed knowledge of Galois fields and Galois rings, please refer to \cite{l2,w1}.
\subsection{Galois Rings of Characteristic $p^{e}$}
Let $p>3$ be a prime and $e\geq2$ be an integer. $({\rm Z}_{p^{e}},+,\times)$ denotes a Galois ring where addition and multiplication
are all modular operations. A monic polynomial $f(t)$ is said to be a basic irreducible polynomial of degree $n$ over ${\rm Z}_{p^{e}}$, if $f(t)$ mod $p$ is a monic irreducible polynomial over ${\rm Z}_{p}$. The Galois ring ${\rm R}_{e,n} =
{\rm GR}(p^{e}, n)$ is the unique extension of degree $n$ over ${\rm Z}_{p^{e}}$ and is isomorphic with ${\rm Z}_{p^{e}}[t]/(f(t))$, where $f(t)$ is a monic basic irreducible polynomial of degree $n$ over ${\rm Z}_{p^{e}}[t]$. ${\rm R}_{e, n}$ is a local ring with unique maximal ideal $(p)=p{\rm R}_{e,n}$, which contains all zero divisors and zeros of ${\rm R}_{e, n}$. The units ${\rm R}^{\times}_{e,n}= {\rm R}_{e,n}\backslash (p)$ are contained in a multiplicative group with the following structure:
$$
{\rm R}^{\times}_{e,n}= G_{1}\times G_{2}
$$
where $G_{1}$ is a cyclic group of order $p^{n}-1$ and $G_{2}$ is a direct product of $n$ cyclic groups each of order $p^{e-1}$.

Define $\Gamma_{e,n}=\{0,1,\xi,\ldots,\xi^{p^{n}-2}\}$ be the Teichm\"{u}ller set in ${\rm R}_{e,n}$, where $\xi\in{\rm R}_{e,n}$ is an nonzero element of order $p^{n}-1$ and $\Gamma^{\times}_{e,n}=\Gamma_{e,n}\backslash \{0\}$. Then $G_{1}=\langle\xi\rangle$ is of order $p^{n}-1$ and $G_{2}=\{1+\theta:\theta\in(p)\}$ is of order $p^{(e-1)n}$.

It can be shown that every element $c\in {\rm R}_{e,n}$ has a unique $p$-adic expansion
$$
c=a_{0}+a_{1}p+\ldots+a_{e-1}p^{e-1}
$$
where $a_{0},a_{1},\ldots,a_{e-1}\in\Gamma_{e,n}$.

Throughout this paper, all the arithmetical operations are in $({\rm R}_{e,n},+,\times)$. For $\alpha\in{\rm R}_{e,n}$, denote ${\rm ord}(\alpha)$ as the order of $\alpha$. $\varphi(n)$, i.e., Euler¡¯s totient function, denotes the number of positive integers which are both less than or equal to the positive integer and coprime with $n$.
\subsection{{\rm IPRNGs} in ${\rm Z}_{p^{e}}$}
In this paper, we study the following IPRNG over Galois rings, which is a direct generalization of the IPRNGs considered in \cite{e8}. Given an arbitrary element $x\in{\rm Z}_{p^{e}}$, the IPRNGs over ${\rm Z}_{p^{e}}$ is
\begin{eqnarray}
\phi(x)=\left\{\begin{array}{cccc}
a x^{-1}+b  &x\in{\rm Z}^{\times}_{p^{e}}\\
b&x\in(p)
\end{array}\right.,
\end{eqnarray}
where $a,b\in {\rm Z}_{p^{e}}$. The initial value associated with (1) is given by $x_{0}\in{\rm Z}_{p^{e}}.$

Set $\phi^{0}(x)=x$ and $\phi^{i+1}=\phi\circ\phi^{i}$ for all $i=0,1,\ldots$. Starting from an initial value $x_{0}\in{\rm Z}_{p^{e}}$, the recurrence $x_{n+1}=\phi^{n}(x_{0})$ $(n=1,2\ldots)$ generates a sequence $x_{0},x_{1},\ldots$ over ${\rm Z}_{p^{e}}$. For every initial value $x_{0}\in{\rm Z}_{p^{e}}$, the smallest integer $L(x_{0},a,b)$ such that $x_{n+L(x_{0},a,b)}=x_{n}$ for all $n\geq n_{0}\geq0$ is called the period of the IPRNGs correspond to $x_{0}$, where $n_{0}$ is a nonnegative integer. Here, we denote $\phi^{r}({\rm Z}_{p^{e}})=\{\phi^{r}(x):x\in {\rm Z}_{p^{e}}\}$ and $\mid\phi^{r}({\rm Z}_{p^{e}})\mid$ be the cardinality of $\phi^{r}({\rm Z}_{p^{e}})$.

The full information on the period distribution is obtained by finding all possible $L(x_{0},a,b)$'s then count the number of a specific $L(x_{0},a,b)$ when $a,b$ and $x_{0}$ traverse all possible elements in ${\rm Z}_{p^{e}}$, where $p>3$ is an odd prime and $e\geq2$ is a integer. The period distribution for $p=2$ and $p=3$ need special analysis.

\section{Period distribution of IPRNGs with $a\in(p)$ in ${\rm Z}_{p^{e}}$}
When $a\in(p)$, the number of IPRNGs is $p^{3e-1}$. It would be better if we have an impression on what the period distribution with $a\in(p)$ looks like. Fig. 1 is a plot of the period distribution of IPRNGs with $a\in(5)$ in ${\rm Z}_{5^3}$. It shows that all the periods are $1$. In the following, the period distribution rules will be worked out analytically.
\begin{figure}[!t]
 \noindent
 \centering\includegraphics[width=3.5in]{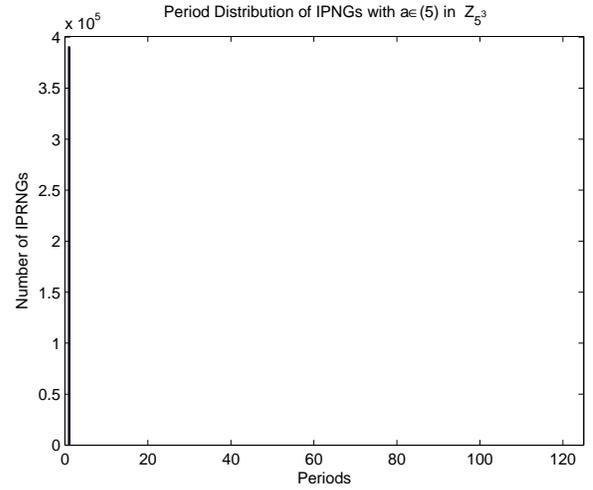}
 \caption{Period distribution of IPRNGs with $a\in(5)$ in ${\rm Z}_{5^{3}}$}\label{fig_sim}
\end{figure}

In the following, we will provide some lemmas, which are necessary to discuss the period distribution of the IPRNGs with $a\in(p)$.
\begin{lemma}
For IPRNG (1) with $a\in(p)$ and $a\neq0$. If $a=c_{k}p^{k}$ where $1\leq k\leq e$ and $c_{k}\in{\rm Z}^{\times}_{p^{e}}$, then $\phi(x)=\phi(x+p^{e-k})$ for all $x\in {\rm Z}_{p^{e}}$ and $x+p^{e-k}\in {\rm Z}_{p^{e}}$.
\end{lemma}
\begin{proof}
The proof is divided into two cases.

Case I: $x\in(p)$ and $x+p^{e-k}\in(p)$. Then, we can get that $\phi(x)=\phi(x+p^{e-k})=b$.

Case II: $x\in{\rm Z}^{\times}_{p^{e}}$ and $x+p^{e-k}\in{\rm Z}^{\times}_{p^{e-k}}$. Then, $x^{-1}\in{\rm Z}^{\times}_{p^{e}}$ and $(x+p^{e-k})^{-1}\in{\rm Z}^{\times}_{p^{e-k}}$. From
$$
x+p^{e-k}\equiv x({\rm mod}p^{e-k}),
$$
we have
$$
(x+p^{e-k})^{-1}\equiv x^{-1}({\rm mod}p^{e-k}),
$$
which implies that
$$
p^{e-k}\mid ((x+p^{e-k})^{-1}-x^{-1}).
$$
Hence,
$$
p^{e}\mid c_{k}p^{k}((x+p^{e-k})^{-1}-x^{-1}),
$$
which means that
$$
c_{k}p^{k}(x+p^{e-k})^{-1}+b\equiv c_{k}p^{k}x^{-1}+b({\rm mod}p^{e}).
$$
Then,
$$
\phi(x+p^{e-k})=\phi(x).
$$

Combining Case I and Case II, we have proven this lemma. The proof is completed.
\end{proof}
\begin{lemma}
For IPRNG (1) with $a\in(p)$ and $a\neq0$, $\cdots\subseteq\phi^{r}({\rm Z}_{p^{e}})\subseteq\cdots\subseteq\phi^{1}({\rm Z}_{p^{e}})\subseteq\phi^{0}({\rm Z}_{p^{e}})$.
\end{lemma}
\begin{proof}
We will prove this lemma by mathematical induction.

Basis: For $r=1$, it is obvious that $\phi^{1}({\rm Z}_{p^{e}})\subseteq\phi^{0}({\rm Z}_{p^{e}})$.

Inductive step: Assume $\phi^{r+1}({\rm Z}_{p^{e}})\subseteq\phi^{r}({\rm Z}_{p^{e}})$ holds for $r\geq1$. Then for any $\phi^{r+2}(x)\in\phi^{r+2}({\rm Z}_{p^{e}})$, we can get that $\phi^{r+1}(x)\in\phi^{r+1}({\rm Z}_{p^{e}})$. Since $\phi^{r+1}({\rm Z}_{p^{e}})\subseteq\phi^{r}({\rm Z}_{p^{e}})$, there exists a $\phi^{r}(x')\in\phi^{r}({\rm Z}_{p^{e}})$, such that $\phi^{r+1}(x)=\phi^{r}(x')$, thus $\phi^{r+2}(x)=\phi^{r+1}(x')\in\phi^{r+1}({\rm Z}_{p^{e}})$. This means that $\phi^{r+2}({\rm Z}_{p^{e}})\subseteq\phi^{r+1}({\rm Z}_{p^{e}})$.

Since both the basis and the inductive step have been proved, it has now been proved by mathematical induction that $\cdots\subseteq\phi^{r}({\rm Z}_{p^{e}})\subseteq\cdots\subseteq\phi^{1}({\rm Z}_{p^{e}})\subseteq\phi^{0}({\rm Z}_{p^{e}})$. The proof is completed.
\end{proof}
\begin{lemma}
For IPRNGs (1) with $a\in(p)$ and $a\neq0$, there exists an integer $r_{0}>0$, such that $|\phi^{r}({\rm Z}_{p^{e}})|<|\phi^{r-1}({\rm Z}_{p^{e}})|$ for all $1\leq r<r_{0}$ and $|\phi^{r}({\rm Z}_{p^{e}})|=1$ for all $r\geq r_{0}$.
\end{lemma}
\begin{proof}
It follows from lemma 2 that $\cdots\leq|\phi^{r}({\rm Z}_{p^{e}})|\leq\cdots\leq|\phi^{1}({\rm Z}_{p^{e}})|\leq|\phi^{0}({\rm Z}_{p^{e}})|$. Since $0<|\phi^{r}({\rm Z}_{p^{e}})|<+\infty$ for all $r\geq0$, there exists an integer $r_{0}>0$, such that $|\phi^{r}({\rm Z}_{p^{e}})|<|\phi^{r-1}({\rm Z}_{p^{e}})|$ for all $1\leq r<r_{0}$ and $|\phi^{r}({\rm Z}_{p^{e}})|$ are equal for all $r\geq r_{0}$.

In the following, we will prove that $|\phi^{r}({\rm Z}_{p^{e}})|=1$ for all $r\geq r_{0}$. Here, we only consider the case that $|\phi^{r}({\rm Z}_{p^{e}})|=2$ for all $r\geq r_{0}$, then the case for $|\phi^{r}({\rm Z}_{p^{e}})|>2$ can be considered similarly.

For $|\phi^{r}({\rm Z}_{p^{e}})|=2$, we assume that there exists a $r'\geq r_{0}$, such that $|\phi^{r'}({\rm Z}_{p^{e}})|=\{x_{1},x_{2}\}$ and $|\phi^{r'+1}({\rm Z}_{p^{e}})|=\{x_{1},x_{2}\}$ with $p^{e-k}\nmid(x_{1}-x_{2})$. If either $x_{1}\in(p)$ or $x_{2}\in(p)$, then contradictions will be easily derived. For $x_{1}\in{\rm Z}^{\times}_{p^{e}}$ and $x_{2}\in{\rm Z}^{\times}_{p^{e}}$, there are two cases.

Case I: $\phi(x_{1})=x_{2}$ and $\phi(x_{2})=x_{1}$. Then, we can get that $\phi(\phi(x_{1}))=\phi(x_{2})$. Thus, there exists an integer $n\neq0$, such that $\phi(x_{1})=x_{2}+n p^{e-k}$. Then, we have $x_{2}=x_{2}+n p^{e-k}$. This is a contradiction.

Case II: $\phi(x_{1})=x_{1}$ and $\phi(x_{2})=x_{2}$. There are two subcases.

Subcase i: $p^{k}\nmid(x_{1}-x_{2})$. It follows from $\phi(x_{1})=ax^{-1}_{1}+b$ and $\phi(x_{2})=ax^{-1}_{1}+b$ that
\begin{eqnarray}
a(x^{-1}_{1}-x^{-1}_{2})=x_{1}-x_{2}.
\end{eqnarray}
For $a=c_{k}p^{k}$, $p^{k}\mid(x_{1}-x_{2})$. This contradicts to $p^{k}\nmid(x_{1}-x_{2})$.

Subcase ii: $p^{k}\mid(x_{1}-x_{2})$. In this case, we assume that $x_{1}-x_{2}=c'_{k}p^{m}$, where $m\geq k$ and $c'_{k}\in{\rm Z}^{\times}_{p^{e}}$. Thus,
$$
x_{1}\equiv x_{2} ({\rm mod }p^{m}).
$$
Then
$$
x^{-1}_{1}\equiv x^{-1}_{2} ({\rm mod }p^{m}).
$$
From (2), we can get that
$$
c_{k}p^{k}(x^{-1}_{1}-x^{-1}_{2})=c'_{k}p^{m},
$$
which means that
\begin{eqnarray}
(x^{-1}_{1}-x^{-1}_{2})=c^{-1}_{k}c'_{k}p^{m-k}.
\end{eqnarray}
From (3), we have
$$
x^{-1}_{1}\nequiv x^{-1}_{2}({\rm mod }p^{m}),
$$

From Subcase i and Subcase ii, we have $\phi(x_{1})=x_{1}$ and $\phi(x_{2})=x_{2}$ lead to a contradiction.

Lemma 3 has been proven by combining Case I and Case II. The proof is completed.
\end{proof}

Now, we are ready to establish our main theorem for period distribution of IPRNGs with $a\in(p)$ on the basis of Lemma 3.
\begin{theorem}
For IPRNGs with $a\in(p)$, the possible periods and the number of each special period are given in Table I.
\end{theorem}
\begin{proof}
Period analysis.

If $a=0$, then it is obvious that $ L(x_{0},a,b)=1$.

If $a\neq0$, then by Lemma 3, we can get that there exits an integer $r_{0}$ such that $\phi^{r+1}(x_{0})=\phi^{r}(x_{0})$ for all $r>r_{0}$. Thus, $L(x_{0};a,b)=1$.

Counting.

When $a$ traverses all elements in $(p)$,  $b$ and $x_{0}$ traverse all elements in ${\rm Z}_{p^{e}}$, respectively, there are $p^{3e-1}$ IPRNGs with $L(x_{0};a,b)=1$. The proof is completed.
\end{proof}
\begin{table}[!t]
\renewcommand{\arraystretch}{1.3}
\caption{Period distribution of IPRNGs with $a\in(p)$ in ${\rm Z}_{p^{e}}$ }
\label{table_example}
\centering
\begin{tabular}{|c|c|}
\hline
\bfseries Period & \bfseries Number of IPRNGs\\
\hline
 \tabincell{c}{$1$} & $p^{3e-1}$\\
\hline
\end{tabular}
\end{table}
\section{Period distribution of IPRNGs with $a\in{\rm Z}^{\times}_{p^{e}}$ and $b\in(p)$ in ${\rm Z}_{p^{e}}$}
When $a\in{\rm Z}^{\times}_{p^{e}}$ and $b\in(p)$, the number of IPRNGs is $(p-1)p^{3e-2}$. It would be better if we have an impression on what the period distribution with $a\in{\rm Z}^{\times}_{p^{e}}$ and $b\in(p)$ looks like. Fig. 2 is a plot of the period distribution of IPRNGs with $a\in{\rm Z}^{\times}_{5^{3}}$ and $b\in(5)$ in ${\rm Z}_{5^3}$. It can be seen from Fig. 2 that the periods distribute very sparsely, some exist and some do not. In the following, the period distribution rules for $a\in{\rm Z}^{\times}_{p^{e}}$ and $b\in(p)$ will be worked out analytically.
\begin{figure}[!t]
 \noindent
 \centering\includegraphics[width=3.5in]{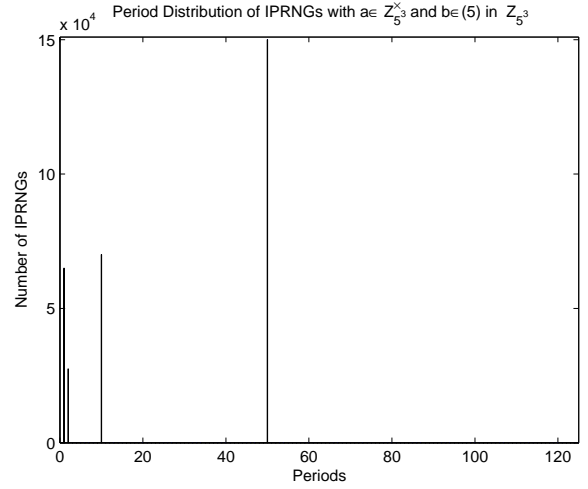}
 \caption{Period distribution of IPRNGs with $a\in {\rm Z}^{\times}_{5^{3}}$ and $b\in(5)$ in ${\rm Z}_{5^{3}}$}\label{fig_sim}
\end{figure}

We rewrite IPRNGs (1) as
\begin{eqnarray}
x_{n+1}=\left\{\begin{array}{cccc}
a x_{n}^{-1}+b &x_{n}\in{\rm Z}^{\times}_{p^{e}}\\
b&x_{n}\in(p)
\end{array}\right.,
\end{eqnarray}
for all $n\geq 0$, where $a\in{\rm Z}^{\times}_{p^{e}}$ and $b\in {\rm Z}_{p^{e}}$.

Hereafter, for presentation convenience, we denote $S(x_{0};a,b)$ as the sequence generated by (4) from initial value $x_{0}$.

In order to get the main results in the rest of this paper, we provide an important lemma which transforms the IPRNGs to 2-dimensional linear feedback shift registers (LFSRs). This lemma is an extensive version of Lemma 1 in \cite{c9}.
\begin{lemma}
Let $a,b,x_{0}\in{\rm Z}_{p^e}$. Define the LFSR
\begin{eqnarray}
y_{n+2}=by_{n+1}+ay_{n},
\end{eqnarray}
for all $n\geq0$, where $y_{0}=1$, $y_{1}=x_{0}$. Then if $m\geq0$ is an integer such that $y_{n}\in {\rm Z}^{\times}_{p^{e}}$ for all $0\leq n\leq m$, then $x_{n}=y_{n+1}y^{-1}_{n}$ for all $0\leq n\leq m$. Moreover, $m$ is the smallest positive integer satisfying $x_{m}\in(p)$ if and only if $m+1$ is the smallest integer satisfying $y_{m+1}\in(p)$.
\end{lemma}
\begin{proof}
We will prove this lemma via mathematical induction.

Basis: For $n=1$, it is obvious that $x_{0}=y_{1}y^{-1}_{0}$.

Inductive step: Assume $x_{k}=y_{k+1}y^{-1}_{k}$ hold for $n=k$, where $0\leq k\leq t-1$. By $y_{k+2}=by_{k+1}+ay_{k}$, we can get that
\begin{eqnarray}
y_{k+2}y^{-1}_{k+1}&=&ay_{k}y^{-1}_{k+1}+b\nonumber\\
&=&a(y_{k+1}y^{-1}_{k})^{-1}+b,\nonumber
\end{eqnarray}
which means that
$$
x_{k+1}=ax^{-1}_{k}+b
$$
Since both the basis and the inductive step have been proved, it has now been proved by mathematical induction that $x_{n}=y_{n+1}y^{-1}_{n}$ for all $0\leq n\leq m$.

By the first assertion of this lemma, we can prove the second assertion. The proof is completed.
\end{proof}

For $a\in{\rm Z}^{\times}_{p^{e}}$ and $b\in(p)$, we provide the a useful lemma, which can be found in \cite{n10}.
\begin{lemma}
IPRNG (4) is a permutation of ${\rm Z}^{\times}_{p^{e}}$ if and only if $a\in{\rm Z}^{\times}_{p^{e}}$ and $b\in(p)$.
\end{lemma}
\begin{remark}
It follows from Lemma 5 that if $a\in{\rm Z}^{\times}_{p^{e}}$ and $b\in(p)$, then $S(x_{0};a,b)$ does not contain any element in $(p)$ for all $x_{0}\in {\rm Z}^{\times}_{p^{e}}$. This situation is quite different from the case $a\in{\rm Z}^{\times}_{p^{e}}$ and $b\in{\rm Z}^{\times}_{p^{e}}$, in which $S(x_{0};a,b)$ may contain elements in $(p)$ for some $x_{0}\in {\rm Z}^{\times}_{p^{e}}$, which will be proved later in Section V. This is the reason why we consider such two cases separately.
\end{remark}

By lemma 5, if $x_{0}\in (p)$, then $x_{n}=b$ for all $n\geq1$. Thus, $L(x_{0};a,b)=1$. In this case, there are $(p-1)p^{e-1}$, $p^{e-1}$ and $p^{e-1}$ choices of $a$, $b$ and $x_{0}$, respectively. Therefore, there are $(p-1)p^{3e-3}$ IPRNGs of period $1$ for this case. In the following, we will analyze period distribution of IPRNGs for the case that $x_{0}\in {\rm Z}^{\times}_{p^{e}}$.

Denote $f(t)=t^{2}-bt-a$ be the characteristic polynomial of recurrent relation (5). Let $\alpha,\beta$ be two roots of $f(t)$, i.e., $f(t)=(t-\alpha)(t-\beta)$. It can be seen that each pair of $a,b$ is uniquely determined by a pair of $\alpha,\beta$. It should be pointed out that $\alpha-\beta$ is always a unit. Actually, it follows from $b\in(p)$ that $p\mid \alpha+\beta$. If $p\mid\alpha-\beta$, then it can be obtained that $p\mid \alpha$. Since $a=\alpha\beta$, it holds that $p\mid a$. This contradicts to $a\in{\rm Z}^{\times}_{p^{e}}$. Then, we can get the general terms of LFSR (5):
\begin{eqnarray}
y_{n}=(\alpha-\beta)^{-1}((x_{0}-\beta)\alpha^{n}+(\alpha-x_{0})\beta^{n}),
\end{eqnarray}
for all $n\geq0$.

By Lemma 4 and (6), we have the following lemma.
\begin{lemma}
If $m\geq0$ is an integer such that $y_{n}\in {\rm Z}^{\times}_{p^{e}}$ for all $0\leq n\leq m$, then $x_{n}=x_{0}$ if and only if
\begin{eqnarray}
(x_{0}-\alpha)(x_{0}-\beta)\alpha^{n}=(x_{0}-\alpha)(x_{0}-\beta)\beta^{n}.\nonumber
\end{eqnarray}
\end{lemma}

On the basis of the above discussions, the period distribution of IPRNGs is analyzed in the following two cases: A. $f(t)$ is reducible in ${\rm Z}_{p^{e}}[t]$; B. $f(t)$ is irreducible in ${\rm Z}_{p^{e}}[t]$ but reducible in its extension ring $ {\rm Z}_{p^{e}}[t]/(f(t))$.

\subsection{$f(t)$ Is Reducible in ${\rm Z}_{p^{e}}[t]$}
In this case, $\alpha,\beta$ are in $ {\rm Z}^{\times}_{p^{e}}$. Let $\alpha=\sum^{e-1}_{i=0}c_{i}p^{i}$, $\beta=\sum^{e-1}_{i=0}d_{i}p^{i}$ and $x_{0}=\sum^{e-1}_{i=0}h_{i}p^{i}$, where $c_{0},d_{0},h_{0}\in{\rm Z}^{\times}_{p}$ and $c_{i},d_{i},h_{i}\in{\rm Z}_{p}$ for all $i=1,2,\ldots,e-1$.

If either $x_{0}-\alpha$ or $x_{0}-\beta$ is a zero. By (6), we have $y_{n}=x_{0}^{n}$ for all $n\geq1$. Thus, $x_{n}=x_{0}$ for all $n\geq1$ , which means that $L(x_{0};a,b)=1$.

As $\alpha$ traverses all elements in ${\rm Z}^{\times}_{p^{e}}$, there are $(p-1)p^{e-1}$ choices of $\alpha$. Once $\alpha$ is chosen, there are $p^{e-1}$ $\beta$'s such that $p\mid\alpha+\beta$. Since each $f(t)$ is uniquely determined by a pair of $\alpha,\beta$, it holds that there are $\frac{(p-1)p^{2e-2}}{2}$ reducible $f(t)$'s in ${\rm Z}_{p^{e}}[t]$, which means that there are $\frac{(p-1)p^{2e-2}}{2}$ pairs of $a,b$. Once $\alpha,\beta$ are chosen, there are two choices of $x_{0}$. Thus, there are $(p-1)p^{2e-2}$ IPRNGs of period $1$.

There are two cases remained 1): both $x_{0}-\alpha$ and $x_{0}-\beta$ are units; 2): one of $x_{0}-\alpha$ and $x_{0}-\beta$ is a zero divisor.
\subsubsection{Both $x_{0}-\alpha$ and $x_{0}-\beta$ are units}
It follows Lemma 6 that $n={\rm ord}(\alpha\beta^{-1})$ is the smallest integer such that Lemma 6 holds. Thus, $L(x_{0};a,b)={\rm ord}(\alpha\beta^{-1})$.

It should be mentioned that $\alpha\beta^{-1}-\alpha^{-1}\beta$ is a zero divisor for this case. Indeed, $\alpha\beta^{-1}-\alpha^{-1}\beta=\alpha^{-1}\beta^{-1}(\alpha-\beta)(\alpha+\beta)$. Since $b\in(p)$ and $b=\alpha+\beta$, it must hold that $p\mid \alpha\beta^{-1}-\alpha^{-1}\beta$.

Now, we are ready to present our results on the period distribution of IPRNGs for this case.

\begin{proposition}
Suppose $f(t)$ is reducible in ${\rm Z}_{p^{e}}[t]$ and $\alpha-\beta$ is a unit. If both $x_{0}-\alpha$ and $x_{0}-\beta$ are units, then the number of IPRNGs of period $2$ is $\frac{(p-3)(p-1)p^{2e-2}}{2}$; the number of IPRNGs of period $2p^{e-k}$ is $\frac{(p-3)(p-1)^{2}p^{3e-k-3}}{2}$, where $1\leq k\leq e-1$,
\end{proposition}
\begin{proof}
Period analysis.

By previous discussion, we have $L(x_{0},a,b)={\rm ord}(\alpha\beta^{-1})$. Let $\alpha\beta^{-1}=\sum^{e-1}_{i=0}a_{i}p^{i}$ and $\alpha^{-1}\beta=\sum^{e-1}_{i=0}b_{i}p^{i}$, where $a_{0},b_{0}\in{\rm Z}^{\times}_{p}$, $a_{i},b_{i}\in{\rm Z}_{p}$ for all $i=1,2,\ldots,e-1$. Since $\alpha\beta^{-1}-\alpha^{-1}\beta$ is a zero divisor, it holds that $a_{0}=b_{0}$. On the other hand, as $\alpha\beta^{-1}=(\alpha^{-1}\beta)^{-1}$, it is valid that $a_{0}b_{0}=1$. Thus, $a_{0}=b_{0}=1$ or $a_{0}=b_{0}=p-1$. Since $b\in(p)$, it holds that $a^{-1}b^{2}+2\equiv p-2 ({\rm mod}p)$, thus, $a_{0}=b_{0}=p-1$.

If $b=0$, then $\alpha+\beta=0$. Thus, $\alpha\beta^{-1}+1=0$, which means that $a_{i}=0$ for all $i=1,2,\ldots,e-1$. Hence, ${\rm ord}(\alpha\beta^{-1})=2$.

If $1\leq k\leq e-1$ is the largest integer such that $p^{k}\mid b$, then $p^{k}\mid \alpha\beta^{-1}+1$. Hence, ${\rm ord}(\alpha\beta^{-1})=2p^{e-k}$.

Counting.

If $b=0$, then the choice of $b$ is unique. If $p^{k}\mid b$, then $a_{i}=0$ for all $i=1,2,\ldots,k-1$, $a_{k}\in{\rm Z}^{\times}_{p}$ and $a_{i}\in{\rm Z}_{p}$ for all $i=k+1,k+2,\ldots,e-1$, there are $(p-1)p^{e-k-1}$ choices of $b$. Once $b$ is chosen, there are $\frac{(p-1)p^{e-1}}{2}$ choices of $a$.

It follows form both $x_{0}-\alpha$ and $x_{0}-\beta$ are units that there are $p-3$ choices of $h_{0}$ and $p$ choices of $h_{i}$ for each $i=1,2,\ldots,e-1$. Thus, for each pair of $a,b$, there are $(p-3)p^{e-1}$ choices of $x_{0}$. Therefore, the number of IPRNGs of period $2$ is $\frac{(p-3)(p-1)p^{2e-2}}{2}$. The number of IPRNGs of period $2p^{e-k}$ is $\frac{(p-3)(p-1)^{2}p^{3e-k-3}}{2}$, where $1\leq k\leq e-1$. The proof is completed.
\end{proof}
\subsubsection{One of $x_{0}-\alpha$ and $x_{0}-\beta$ is a zero divisor}
In this case, $x_{0}-\alpha$ and $x_{0}-\beta$ can not both be zero divisors. Without loss of generality, we suppose $x_{0}-\alpha$ is a zero divisor and $x_{0}-\beta$ is not. Let $1\leq k\leq e-1$ be the largest integer such that $p^{k}\mid x_{0}-\alpha$, then by Lemma 6, we have $p^{k}(\alpha\beta^{-1})^{n}=p^{k}$. Let $(\alpha\beta^{-1})^{n}=\sum^{e-1}_{i=0}g_{i}p^{i}$, where $g_{i}\in {\rm Z}_{p}$, then we have
$$
p^{k}(g_{0}+g_{1}p^{1}+\ldots+g_{e-s-1}p^{e-k-1})=p^{k},
$$
which means that
\begin{eqnarray}
g_{0}+g_{1}p+\ldots+g_{e-k-1}p^{e-k-1}=1.
\end{eqnarray}

Define $\eta^{e}_{k}$ be a reduction map from ${\rm Z}_{p^{e}}$ to ${\rm Z}_{p^{e-k}}$, then we have $n={\rm ord}(\eta^{e}_{k}(\alpha\beta^{-1}))$ is the smallest integer such that (7) holds, which means that $L(x_{0};a,b)={\rm ord}(\eta^{e}_{k}(\alpha\beta^{-1}))$.

Now, we are ready to present our results on the period distribution of the IPRNGs in this case.
\begin{proposition}
Suppose $f(t)$ is reducible in ${\rm Z}_{p^{e}}[t]$ and $\alpha-\beta$ is a unit. If one of $x_{0}-\alpha$ and $x_{0}-\beta$ is a zero divisor, then the number of IPRNGs of period $2$ is $((e-1)p-e+1)(p-1)p^{2e-2}$; the number of IPRNGs of period $2p^{e-k-s}$ is $(p-1)^{3}p^{3e-k-s-3}$, where $1\leq k\leq e-1$ and $1\leq s\leq e-k-1$.
\end{proposition}
\begin{proof}
Period analysis.

Let $\alpha\beta^{-1}=\sum^{e-1}_{i=0}a_{i}p^{i}$, $a_{i}\in {\rm Z}_{p}$ for all $i=0,1,\ldots,e-1$.

If $b=0$, then $\alpha+\beta=0$, we have $\alpha\beta^{-1}+1=0$. Thus, $a_{0}=p-1$, $a_{i}=0$ for all $i=1,2,\ldots,e-1$. Then ${\rm ord}(\eta^{e}_{k}(\alpha\beta^{-1}))=2$.

If $e-k\leq s\leq e-1 $ is the largest integer such that $p^{s}\mid b$, then $p^{s}\mid\alpha\beta^{-1}+1$, thus, $\eta^{e}_{k}(\alpha\beta^{-1})=a_{0}$, which means that ${\rm ord}(\eta^{e}_{k}(\alpha\beta^{-1}))=2$.

If $1\leq s\leq e-k-1 $ is the largest integer such that $p^{s}\mid b$, then $p^{s}\mid\alpha\beta^{-1}+1$, thus, $\eta^{e}_{k}(\alpha\beta^{-1})=a_{0}+\sum^{e-k-1}_{i=s}a_{i}p^{i}$, where $a_{s}\in {\rm Z}^{\times}_{p}$ and $a_{i}\in {\rm Z}_{p}$ for all $i=s+1,s+2,\ldots,e-k-1$. Then ${\rm ord}(\eta^{e}_{k}(\alpha\beta^{-1}))=2p^{e-k-s}$.

Counting.

If $L(x_{0};a,b)=2$, then either $b=0$ or $p^{s}\mid b$, where $e-k\leq s\leq e-1$.

As $b=0$, there are $\frac{(p-1)p^{e-1}}{2}$ choices of $a$ and  $2p^{e-1}$ choices of $x_{0}$.

As $p^{s}\mid b$, there are $(p-1)p^{e-s-1}$ choices of $b$. Once $b$ is chosen, there are $\frac{(p-1)p^{e-1}}{2}$ choices of $a$. Since $p^{k}\mid x_{0}-\alpha$ or $p^{k}\mid x_{0}-\beta$, there are $2(p-1)p^{e-k-1}$ choices of $x_{0}$ altogether.

Thus, the number of IPRNGs of period $2$ is
\begin{eqnarray}
&&(p-1)p^{2e-2}+\sum^{e-1}_{k=1}\sum^{e-1}_{s=e-k}(p-1)^{3}p^{3e-k-s-3}\nonumber\\
&&=((e-1)p-e+1)(p-1)p^{2e-2}\nonumber
\end{eqnarray}

If $L(x_{0};a,b)=2p^{e-k-s}$, then $p^{s}\mid b$, where $1\leq s\leq e-k-1$.

As $p^{s}\mid b$, there are $(p-1)p^{e-s-1}$ choices of $b$. Once $b$ is chosen, there are $\frac{(p-1)p^{e-1}}{2}$ choices of $a$. Since $p^{k}\mid x_{0}-\alpha$ or $p^{k}\mid x_{0}-\beta$, there are $2(p-1)p^{e-k-1}$ choices of $x_{0}$.

Thus, the number of IPRNGs of period $2p^{e-k-s}$ is $(p-1)^{3}p^{3e-k-s-3}$, where $1\leq k\leq e-1$ and $1\leq s\leq e-k-1$. The proof is completed.
\end{proof}
\subsection{$f(t)$ Is Irreducible in ${\rm Z}_{p^{e}}[t]$}
In this case, $f(t)$ must be reducible in ${\rm Z}_{p^{e}}[t]/(f(t))$. Since $p\nmid\alpha-\beta$, it is valid that $t-\alpha$ and $t-\beta$ are coprime in ${\rm Z}_{p}$. Thus, by the Hensel's lemma in \cite{w1}, we can get that $f(t)$ is a basic irreducible polynomial in ${\rm Z}_{p}$. Therefore, ${\rm Z}_{p^{e}}[t]/(f(t))$ is a Galois ring which is isomorphic with ${\rm R}_{e,2}$. When $a$ traverses all elements in ${\rm Z}^{\times}_{p^{e}}$ and $b$ traverses all elements in $(p)$, there are $(p-1)p^{2e-2}$ $f(t)$'s in ${\rm Z}_{p^{e}}[t]$. In case A, we obtain that there are $\frac{(p-1)p^{2e-2}}{2}$ $f(t)$'s which are reducible in ${\rm Z}_{p^{e}}[t]$. Thus, there are $\frac{(p-1)p^{2e-2}}{2}$ $f(t)$'s which are irreducible in ${\rm Z}_{p^{e}}[t]$, which means that there are $\frac{(p-1)p^{2e-2}}{2}$ pairs of $a,b$ such that $f(t)$ is irreducible in ${\rm Z}_{p^{e}}[t]$.

Since $\alpha,\beta\in {\rm R}_{e,2}$ but $\alpha,\beta\notin {\rm Z}_{p^{e}}$, it is valid that both $x_{0}-\alpha$ and $x_{0}-\beta$ are units for all $x_{0}\in{\rm Z}^{\times}_{p^{e}}$. Then, it follows from Lemma 6 that $L(x_{0};a,b)={\rm ord}(\alpha\beta^{-1})$.

We present the following proposition without proof because the proof is the same as Proposition 1.
\begin{proposition}
Suppose $f(t)$ is irreducible in ${\rm Z}_{p^{e}}[t]$. Then the number of IPRNGs of $2$ is $\frac{(p-1)^{2}p^{2e-2}}{2}$. The number of IPRNGs of period $2p^{e-k}$ is $\frac{(p-1)^{3}p^{3e-k-3}}{2}$, where $1\leq k\leq e-1$,.
\end{proposition}
\begin{proof}
Period analysis.

By previous discussion, we have $L(x_{0},a,b)={\rm ord}(\alpha\beta^{-1})$. Let $\alpha\beta^{-1}=\sum^{e-1}_{i=0}a_{i}p^{i}$ and $\alpha^{-1}\beta=\sum^{e-1}_{i=0}b_{i}p^{i}$, where $a_{i},b_{i}\in \Gamma_{e,2}$ for all $i=1,2,\ldots,e-1$. Since $b\in (p)$, it is valid that $p\mid \alpha+\beta$. Thus, $p\mid\alpha\beta^{-1}-\alpha^{-1}\beta$, which means that $a_{0}=b_{0}$. On the other hand, $(\alpha\beta^{-1})(\alpha^{-1}\beta)=1$, then $a_{0}b_{0}=1$. Thus, $a_{0}^{2}=1$ which means that ${\rm ord}(a_{0})=2$.

If $b=0$, then $\alpha+\beta=0$. Thus, $\alpha\beta^{-1}+1=0$, which means that ${\rm ord}(\alpha\beta^{-1})=2$.

If $1\leq k\leq e-1$ is the largest integer such that $p^{k}\mid b$, then $p^{k}\mid \alpha\beta^{-1}+1$. Hence, ${\rm ord}(\alpha\beta^{-1})=2p^{e-k}$.

Counting.

If $b=0$, then the choice of $b$ is unique. If $p^{k}\mid b$, then there are $(p-1)p^{e-k-1}$ choices of $b$. Once $b$ is chosen, there are $\frac{(p-1)p^{e-1}}{2}$ choices of $a$. Since both $x_{0}-\alpha$ and $x_{0}-\beta$ are units, there are $(p-1)p^{e-1}$ choices of $x_{0}$. Therefore, the number of IPRNGs of period $2$ is $\frac{(p-1)^{2}p^{2e-2}}{2}$. The number of IPRNGs of period $2p^{e-k}$ is $\frac{(p-1)^{3}p^{3e-k-3}}{2}$, where $1\leq k\leq e-1$. The proof is completed.
\end{proof}

Now, we have discussed all cases for the period distribution of IPRNGs with $a\in{\rm Z}^{\times}_{p^{e}}$ and $b\in(p)$. The overall results are summarized in the following theorem.
\begin{theorem}
For IPRNGs with $a\in{\rm Z}^{\times}_{p^{e}}$ and $b\in(p)$, the possible periods and the number of each special period are given in Table II.
\end{theorem}
\begin{table}[!t]
\renewcommand{\arraystretch}{1.3}
\caption{Period distribution of IPRNGs with $a\in{\rm Z}^{\times}_{p^{e}}$ and $b\in(p)$ in ${\rm Z}_{p^{e}}$}
\label{table_example}
\centering
\begin{tabular}{|c|c|}
\hline
\bfseries Periods & \bfseries Number of IPRNGs\\
\hline
 \tabincell{c}{$1$} & $(p-1)(p^{3e-3}+p^{2e-2})$\\
\hline
\tabincell{c}{$2$} & \tabincell{c}{$(ep-e-1)(p-1)p^{2e-2}$}\\
\hline
\tabincell{c}{$2p^{e-k}$ \\ for each $1\leq k\leq e-1$} & $(p-1)^{2}(p-2)p^{3e-k-3}$\\
\hline
\tabincell{c}{$2p^{e-k-s}$ \\ for each $1\leq k\leq e-1$ and $1\leq s\leq e-k-1$} & $(p-1)^{3}p^{3e-k-s-3}$\\
\hline
\end{tabular}
\end{table}
\begin{example}
The following example is given to compare the theoretical and experimental results. A computer program has been written to exhaust all possible IPRNGs with $a\in{\rm Z}^{\times}_{5^{3}}$ and $b\in(5)$ in ${\rm Z}_{5^{3}}$ to find the period by brute force, the results are shown in Fig. 2.

Table III lists the complete result we have obtained. It provides full information on the period distribution of the IPRNGs. The maximal period is $50$ while the minimal period is $1$. As it is shown in Fig. 2 and Table III, the theoretical and experimental results fit well. The analysis process also indicates how to choose the parameters and the initial values such that the IPRNGs fit specific periods.
\end{example}
\begin{table}[!t]
\renewcommand{\arraystretch}{1.3}
\caption{Period distribution of IPRNGs with $a\in{\rm Z}^{\times}_{5^{3}}$ and $b\in(5)$ in ${\rm Z}_{5^{3}}$}
\label{table_example}
\centering
\begin{tabular}{|c|c|c|c|c|}
\hline
Periods &1 &2 &10&50\\
\hline
Number of IPRNGs &65000&27500&70000&150000\\
\hline
\end{tabular}
\end{table}
\section{Period distribution of IPRNGs with $a\in{\rm Z}^{\times}_{p^{e}}$ and $b\in{\rm Z}^{\times}_{p^{e}}$ in ${\rm Z}_{p^{e}}$}
When $a\in{\rm Z}^{\times}_{p^{e}}$ and $b\in{\rm Z}^{\times}_{p^{e}}$, the number of IPRNGs is $(p-1)^{2}p^{3e-2}$. It would be better if we have an impression on what the period distribution with $a\in{\rm Z}^{\times}_{p^{e}}$ and $b\in(p)$ looks like. Fig. 3 is a plot of the period distribution of IPRNGs (1) with $a\in{\rm Z}^{\times}_{5^{3}}$ and $b\in{\rm Z}^{\times}_{5^{3}}$ in ${\rm Z}_{5^3}$.  It can be seen from Fig. 3 that the periods distribute very sparsely, some exist and some do not. In the following, the period distribution rules for $a\in{\rm Z}^{\times}_{p^{e}}$ and $b\in{\rm Z}^{\times}_{p^{e}}$ will be worked out analytically.
\begin{figure}[!t]
 \noindent
 \centering\includegraphics[width=3.5in]{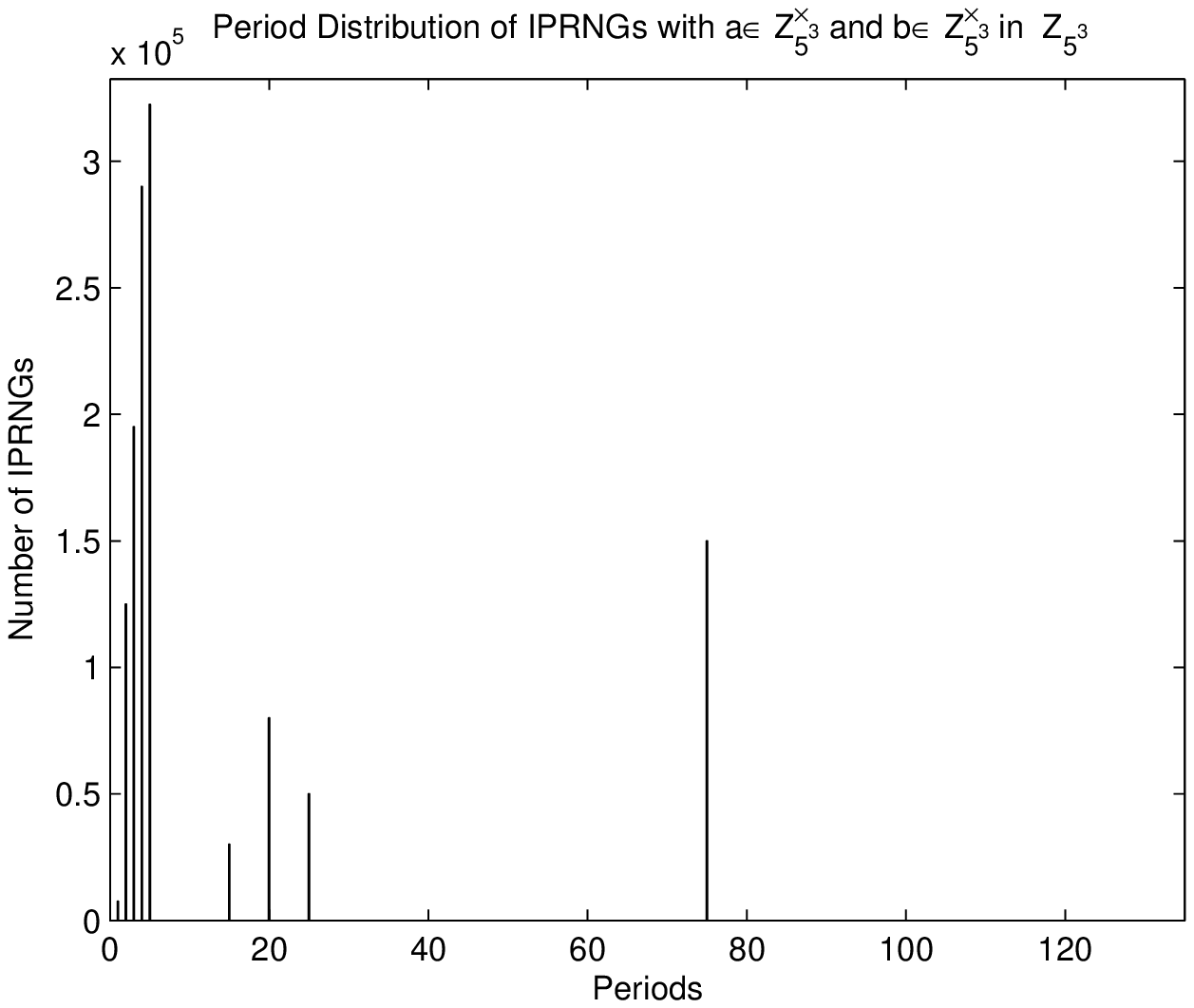}
 \caption{Period distribution of IPRNGs with $a\in {\rm Z}^{\times}_{5^{3}}$ and $b\in {\rm Z}^{\times}_{5^{3}}$ in ${\rm Z}_{5^{3}}$}\label{fig_sim}
\end{figure}

It follows from (6) that if $\alpha-\beta$ is a unit, then we are able to obtain the general term of LFSRs (5). Otherwise, we can not get that its general term. Thus, the period distribution of the IPRNGs is analyzed in the following two cases: A. $\alpha-\beta$ is a unit; B. $\alpha-\beta$ is a zero or a zero divisor, where $\alpha,\beta$ are roots of $f(t)$.

\subsection{$\alpha-\beta$ is a unit}
In this case, if $f(t)$ is reducible in ${\rm Z}_{p^{e}}[t]$, then $\alpha,\beta\in {\rm Z}^{\times}_{p^{e}}$. If $f(t)$ is irreducible in ${\rm Z}_{p^{e}}[t]$, then $f(t)$ must reducible in its extension ring ${\rm Z}_{p^{e}}[t]/(f(t))$. In the following, we will consider the two subcases, 1): $f(t)$ is reducible in ${\rm Z}_{p^{e}}[t]$; 2): $f(t)$ is irreducible in ${\rm Z}_{p^{e}}[t]$ but reducible in its extension ring ${\rm Z}_{p^{e}}[t]/(f(t))$. In both subcases, a pair of $a,b$ is uniquely determined by a pair of $\alpha,\beta$.
\subsubsection{$f(t)$ is reducible in ${\rm Z}_{p^{e}}[t]$}
Let $\alpha=\sum^{e-1}_{i=0}c_{i}p^{i}$, $\beta=\sum^{e-1}_{i=0}d_{i}p^{i}$ and $x_{0}=\sum^{e-1}_{i=0}h_{i}p^{i}$, where $c_{i},d_{i},h_{i}\in{\rm Z}_{p}$ then it follows from $\alpha-\beta$ is a unit that $c_{0}\neq d_{0}$. For presentation convenience, we denote $c_{0}=\omega_{1},d_{0}=\omega_{2}$ and $h_{0}=\pi$.

It follows from recurrence relation (5) that
$$
y_{n+2}\equiv (\omega_{1}+\omega_{2}) y_{n+1}-\omega_{1}\omega_{2}y_{n} ({\rm mod}p).
$$
Let $x'_{n}=x_{n} {\rm mod}p$ and $y'_{n}=y_{n} {\rm mod}p$ for all $n=1,2,\ldots$. Then, we obtain
\begin{eqnarray}
y'_{n+2}= (\omega_{1}+\omega_{2}) y'_{n+1}-\omega_{1}\omega_{2}y'_{n}.
\end{eqnarray}

Similar to (6), we have the general term of (8)
$$
y'_{n}=(\omega_{1}-\omega_{2})^{-1}((\pi-\omega_{2})\omega^{n}_{1}+(\omega_{1}-\pi)\omega^{n}_{2}).
$$

If both $\pi-\omega_{1}\neq 0$ and $\pi-\omega_{2}\neq 0$, then $y'_{n}=0$ if and only if
\begin{eqnarray}
(\omega_{1}\omega_{2}^{-1})^{n}=(\pi-\omega_{1})(\pi-\omega_{2})^{-1}.
\end{eqnarray}

For presentation convenience, we denote $\Omega=\{\omega_{1}\omega_{2}^{-1},(\omega_{1}\omega_{2}^{-1})^{2},\ldots,(\omega_{1}\omega_{2}^{-1})^{{\rm ord}(\omega_{1}\omega_{2}^{-1})-1}\}$.

If $(\pi-\omega_{1})(\pi-\omega_{2})^{-1}\in\Omega$, there exists $1\leq n\leq p-1$ such that (9) holds, thus, $S(x_{0};a,b)$ must contains some elements in $(p)$; if $(\pi-\omega_{1})(\pi-\omega_{2})^{-1}\notin\Omega$, there does not exist any $n$ such that (9) holds, thus, $S(x_{0};a,b)$ does not contain any element in $(p)$.

On the other hand, if either $\pi-\omega_{1}=0$ or $\pi-\omega_{2}=0$, then $y'_{n}\neq 0$ for all $n=1,2,\ldots$, which means that $S(x_{0};a,b)$ does not contain any element in $(p)$.

Now, we are ready to present our results on the period distribution of IPRNGs for this case.
\begin{proposition}
Suppose $f(t)$ is reducible in ${\rm Z}_{p^{e}}[t]$ and $\alpha-\beta$ is a unit. If $(\pi-\omega_{1})(\pi-\omega_{2})^{-1}\in\Omega$, then $L(x_{0};a,b)$ traverses the set $\{k-1:k>2,k\mid p-1\}$. For each $k$, there are $(k-1)(p-1)p^{2e-2}\sum^{e-1}_{i=0}\frac{\varphi(kp^{i})}{2}$ IPRNGs of period $k-1$.
\end{proposition}
\begin{proof}
Period analysis.

Since $b\in{\rm Z}_{p^e}$ and $b=\alpha+\beta$, it holds that $p\nmid \alpha+\beta$. Combining $p\nmid \alpha-\beta$, we have $p\nmid \alpha\beta^{-1}-\alpha^{-1}\beta$, which means that $\omega_{1}\omega^{-1}_{2}\neq \omega_{1}^{-1}\omega_{2}$. Again, since $(\omega_{1}\omega^{-1}_{2})(\omega_{1}^{-1}\omega_{2})=1$, it is valid that $\omega_{1}\omega^{-1}_{2}\neq 1,p-1$ which means that ${\rm ord}(\omega_{1}\omega^{-1}_{2})>2$.

If $(\pi-\omega_{1})(\pi-\omega_{2})^{-1}\in\Omega$, then $S(x_{0};a,b)$ contains some elements in $(p)$. Thus, $L(x_{0};a,b)=L(b;a,b)$. Then, we consider the case that $x_{0}=b$, which means that $\pi=\omega_{1}+\omega_{2}$. By (9), we have $y'_{n}=0$ if and only if $(\omega_{1}\omega_{2}^{-1})^{n+1}=1$. Thus, $n'={\rm ord}(\omega_{1}\omega^{-1}_{2})-1$ is the smallest integer such that $y'_{n'}=0$. By Lemma 4, we have $x'_{n'-1}=0$, thus, $x_{n'}=b$, which means that $L(x_{0};a,b)={\rm ord}(\omega_{1}\omega^{-1}_{2})-1$. Since $\omega_{1}\omega^{-1}_{2}\in {\rm Z}_{p}$, it holds that ${\rm ord}(\omega_{1}\omega^{-1}_{2})\mid p-1$. Hence, $L(x_{0};a,b)$ traverses the set $\{k-1:k>2,k\mid p-1\}$.

Counting.

For $L(x_{0};a,b)=k-1$, there are $k-1$ $\pi$'s such that $(\pi-\omega_{1})(\pi-\omega_{2})^{-1}\in\Omega$ and $p$ choices of $h_{i}$ for all $i=1,2,\ldots,e-1$. Thus, there are $(k-1)p^{e-1}$ choices of $x_{0}$.

Since $\alpha$ and $\beta$ are roots of $f(t)$, it can be verified that $\alpha\beta^{-1}$ and $\alpha^{-1}\beta$ are roots of $g(t)=t^{2}+(a^{-1}b^{2}+2)t+1$. Therefore, $a^{-1}b^{2}+2=\alpha\beta^{-1}+\alpha^{-1}\beta$. Thus, $a=b^{2}(\alpha\beta^{-1}+\alpha^{-1}\beta-2)$. Since $f(t)$ is reducible in ${\rm Z}_{p^{e}}[t]$, it is valid that ${\rm ord}(\alpha\beta^{-1})=kp^{i}$, where $1\leq i\leq e-1$. For each $kp^{i}$, there are $\varphi(kp^{i})$ elements whose order is $kp^{i}$ and there are $\frac{\varphi(kp^{i})}{2}$ different $\alpha\beta^{-1}+\alpha^{-1}\beta-2$ 's. Thus, there are $\sum^{e-1}_{i=0}\frac{\varphi(kp^{i})}{2}$ choices of $\alpha\beta^{-1}+\alpha^{-1}\beta-2 $.

As a result of ${\rm ord}(\omega_{1}\omega^{-1}_{2})>2$, we have $\alpha\beta^{-1}+\alpha^{-1}\beta-2$ is a unit. The number of choices of $b$ is $(p-1)p^{e-1}$. Once $b$ and $\alpha\beta^{-1}+\alpha^{-1}\beta-2$ are chosen, $a$ is uniquely determined. Hence, for each $k$, there are $(k-1)(p-1)p^{2e-2}\sum^{e-1}_{i=0}\frac{\varphi(kp^{i})}{2}$ IPRNGs of period $k-1$. The proof is completed.
\end{proof}
\begin{proposition}
Suppose $f(t)$ is reducible in ${\rm Z}_{p^{e}}$. If $(\pi-\omega_{1})(\pi-\omega_{2})^{-1}\notin\Omega$, then $L(x_{0};a,b)$ traverses the set $\{k=k_{1}k_{2}:2<k_{1}<p-1,k_{1}\mid p-1,k_{2}\mid p^{e-1}\}$. For each $k$, there are $(p-(k_{1}-1))(p-1)p^{2e-2}\frac{\varphi(k)}{2}$ IPRNGs of period $k$
\end{proposition}
\begin{proof}
Period analysis.

By recurrence relation (5), we can get that $x_{n}=x_{0}$ if and only if $(\alpha\beta^{-1})^{-1}=1$. Thus, $L(x_{0};a,b)={\rm ord}(\alpha\beta^{-1})$. In Proposition 4, we have proven that ${\rm ord}(\omega_{1}\omega^{-1}_{2})>2$. Since $\alpha\beta^{-1} \in {\rm Z}^{\times}_{p^{e}}$, it is valid that ${\rm ord}(\alpha\beta^{-1})$ traverse the set $\{k=k_{1}k_{2}:k_{1}>2,k_{1}\mid p-1,k_{2}\mid p^{e-1}\}$.

Since $(\pi-\omega_{1})(\pi-\omega_{2})^{-1}\notin \Omega$, it is valid that $S(x_{0};a,b)$ does not contain any element in $(p)$. Thus, $\omega_{1}\omega_{2}^{-1}$ is not a primitive element in ${\rm Z}_{p}$, which means that ${\rm ord}(\omega_{1}\omega_{2}^{-1})\neq p-1$. Hence, $L(x_{0};a,b)$ traverses the set $\{k=k_{1}k_{2}:2<k_{1}<p-1,k_{1}\mid p-1,k_{2}\mid p^{e-1}\}$.

Counting.

For each $L(x_{0};a,b)=k$, there are $p-(k_{1}-1)$ $\pi$'s such that $(\pi-\omega_{1})(\pi-\omega_{2})^{-1}\notin\Omega$ and $p$ choices of $h_{i}$ for all $i=1,2,\ldots,e-1$. Thus, there are $(p-(k_{1}-1))p^{e-1}$ choices of $x_{0}$.

The rest of the proof is the same as which in in Proposition 4, thus, we omit it. Finally, we have for each $k$, there are $(p-(k_{1}-1))(p-1)p^{2e-2}\frac{\varphi(k)}{2}$ IPRNGs of period $k$. The proof is completed.
\end{proof}
\begin{proposition}
Suppose $f(t)$ is reducible in ${\rm Z}_{p^{e}}[t]$ and $\alpha-\beta$ is a unit. If either $x_{0}-\alpha$ or $x_{0}-\beta$ is a zero, then $L(x_{0};a,b)=1$. There are $(p-3)(p-1)p^{2e-2}$ of period $1$. If either $x_{0}-\alpha$ or $x_{0}-\beta$ is a zero divisor, then $L(x_{0};a,b)$ traverses the set $\{k=k_{1}k_{2}:k_{1}>2,k_{1}\mid p-1,k_{2}\mid p^{e-k_{3}-1},1\leq k_{3} \leq e-1\}$. For each $k_{3}$, there are $\varphi(k)(p-1)^{2}p^{2e-2}$ IPRNGs of period $k$.
\end{proposition}
\begin{proof}
Period analysis.

If either $x_{0}-\alpha$ or $x_{0}-\beta$ is a zero, then $y_{n}=x^{n}_{0}$. Thus, $x_{n}=x_{0}$ for all $n=1,2,\ldots$, which means that $L(x_{0};a,b)=1$.

If either $x_{0}-\alpha$ or $x_{0}-\beta$ is a zero divisor, we suppose $1\leq k_{3}\leq e-1$ is the largest integer such that $p^{k_{3}}\mid x_{0}-\alpha$ or $p^{k_{3}}\mid x_{0}-\beta$, then we can get that $L(x_{0};a,b)=\eta^{e}_{k_{3}}(\alpha\beta^{-1})$. Thus, $L(x_{0};a,b)$ traverses the set $\{k=k_{1}k_{2}:k_{1}>2,k_{1}\mid p-1,k_{2}\mid p^{e-k_{3}-1},1\leq k_{3} \leq e-1\}$.

Counting.

For $L(x_{0};a,b)=1$, $\alpha,\beta$ traverses all suitable elements in ${\rm Z}_{p^{e}}$, i.e. both $\alpha-\beta$ and $\alpha+\beta$ are units, there are $\frac{(p-3)(p-1)p^{2e-2}}{2}$ pairs of $\alpha,\beta$. Once $\alpha,\beta$ are chosen, there are $2$ choices of $x_{0}$. Thus, there are $(p-3)(p-1)p^{2e-2}$ IPRNGs of period $1$.

For $L(x_{0};a,b)=k$, since either $p^{k_{3}}\mid x_{0}-\alpha$ or $p^{k_{3}}\mid x_{0}-\beta$, it is valid that $\pi=\omega_{1}$ or $\pi=\omega_{2}$ and $p-1$ choices of $h_{k_{3}}$ for all $i=k_{3}+1,k_{3}+2,\ldots,e-1$. Thus, there are $2(p-1)p^{e-k_{3}-1}$ choices of $x_{0}$ altogether.

Let $\alpha\beta^{-1}=\sum^{e-k_{3}-1}_{i=0}a_{i}p^{i}+\sum^{e-1}_{i=e-k_{3}}a_{i}p^{i}$. Since $\eta^{e}_{k_{3}}(\alpha\beta^{-1})=k$, there are $\varphi(k)$ choices of $\eta^{e}_{k_{3}}(\alpha\beta^{-1})$. Once $\eta^{e}_{k_{3}}(\alpha\beta^{-1})$ is chosen, which means that $a_{i}$ for all $i=0,1,\ldots,e-k_{3}-1$ are chosen, there are $p$ choices of $a_{i}$ for all $i=e-k_{3},e-k_{3}+1,\ldots,e-1$. Thus, there are $\varphi(k)p^{k_{3}}$ choices of $\alpha\beta^{-1}$. Then, there are $\frac{\varphi(k)p^{k_{3}}}{2}$ different $\alpha\beta^{-1}+\alpha^{-1}\beta-2$ 's. The number of choices of $b$ is $(p-1)p^{e-1}$. Once $b$ and $\alpha\beta^{-1}+\alpha^{-1}\beta-2$ are chosen, $a$ is uniquely determined by $b^{2}(\alpha\beta^{-1}+\alpha^{-1}\beta-2)^{-1}$. Hence, for each $k$, there are $\varphi(k)(p-1)^{2}p^{2e-2}$ IPRNGs of period $k$. The proof is completed.
\end{proof}
\subsubsection{$f(t)$ is irreducible in ${\rm Z}_{p^{e}}[t]$}
In this case, $f(t)$ must be reducible in ${\rm Z}_{p^{e}}[t]/(f(t))$. Since $p\nmid\alpha-\beta$, it is valid that $t-\alpha$ and $t-\beta$ are coprime in ${\rm Z}_{p}$. Thus, by the Hensel's lemma in \cite{w1}, we can get that $f(t)$ is a basic irreducible polynomial in ${\rm Z}_{p}$. Therefore, ${\rm Z}_{p^{e}}[t]/(f(t))$ is a Galois ring which is isomorphic with ${\rm R}_{e,2}$.

Let $\alpha=\sum^{e-1}_{i=0}c_{i}p^{i}$, $\beta=\sum^{e-1}_{i=0}d_{i}p^{i}$ and $x_{0}=\sum^{e-1}_{i=0}h_{i}p^{i}$, where $c_{i},d_{i}\in\Gamma_{e,2}$ and $h_{i}\in{\rm Z}_{p}$ for all $i=0,1,\ldots,e-1$, then it follows from $\alpha-\beta$ is a unit that $c_{0}\neq d_{0}$.

For presentation convenience, we also denote $c_{0}=\omega_{1},d_{0}=\omega_{2}$ and $h_{0}=\pi$.

Since both $\alpha$ and $\beta$ are not in ${\rm Z}_{p^{e}}$, it is valid that both $x_{0}-\alpha$ and $x_{0}-\beta$ are units, which means that both $\pi-\omega_{1}$ and $\pi-\omega_{2}$ are units. As it is discussed in Case A, we can get that if $(\pi-\omega_{1})(\pi-\omega_{2})^{-1}\in\Omega$, then $S(x_{0};a,b)$ must contain some elements in $(p)$; if $(\pi-\omega_{1})(\pi-\omega_{2})^{-1}\notin\Omega$, then $S(x_{0};a,b)$ does not contain any element in $(p)$.

Now, we are ready to present our results on the period distribution of IPRNGs for this case.
\begin{proposition}
Suppose $f(t)$ is irreducible in ${\rm Z}_{p^{e}}[t]$ and $p\nmid \alpha-\beta$. If $(\pi-\omega_{1})(\pi-\omega_{2})^{-1}\in\Omega$, then $L(x_{0};a,b)$ traverses the set $\{k-1:k>2,k\mid p+1\}$. For each $k$, there are $(k-1)(p-1)p^{2e-2}\sum^{e-1}_{i=0}\frac{\varphi(kp^{i})}{2}$ IPRNGs of period $k-1$.
\end{proposition}
\begin{proof}
Period analysis.

Since $b\in{\rm Z}^{\times}_{p^e}$ and $b=\alpha+\beta$, it holds that $p\nmid \alpha+\beta$. Combining $p\nmid \alpha-\beta$, we have $p\nmid \alpha\beta^{-1}-\alpha^{-1}\beta$, which means that $\omega_{1}\omega^{-1}_{2}\neq \omega_{1}^{-1}\omega_{2}$. Since $(\omega_{1}\omega^{-1}_{2})(\omega_{1}^{-1}\omega_{2})=1$, it is valid that ${\rm ord}(\omega_{1}\omega^{-1}_{2})>2$.

Since $(\pi-\omega_{1})(\pi-\omega_{2})^{-1}\in\Omega$, it is valid that $S(x_{0};a,b)$ contains some elements in $(p)$. Thus, $L(x_{0};a,b)=L(b;a,b)$. Thus, $L(x_{0};a,b)={\rm ord}(\omega_{1}\omega^{-1}_{2})-1$. Since $\omega_{1}\omega^{-1}_{2}\in \Gamma_{e,2}$, it must hold that ${\rm ord}(\omega_{1}\omega^{-1}_{2})\mid p^{2}-1$. Notice that $\alpha\beta^{-1}\notin{\rm Z}_{p^{e}} $. Since ${\rm Z}_{p^{e}}\subseteq {\rm Z}_{p^{e}}[t]/(f(t))$, it is valid that all units in ${\rm Z}_{p^{e}}$ are contained in ${\rm Z}_{p^{e}}[t]/(f(t))$, which means that ${\rm ord}(\omega_{1}\omega^{-1}_{2})\nmid p-1$. Thus, ${\rm ord}(\omega_{1}\omega^{-1}_{2})\mid p+1$. Hence, $L(x_{0};a,b)$ traverses the set $\{k-1:k>2,k\mid p+1\}$.

Counting.

For $L(x_{0};a,b)=k$, there are $k-1$ $\pi$'s such that $(\pi-\omega_{1})(\pi-\omega_{2})^{-1}\in\Omega$ and $p$ choices of $h_{i}$ for all $i=1,2,\ldots,e-1$. Thus, there are $(k-1)p^{e-1}$ choices of $x_{0}$.

Since $\alpha$ and $\beta$ are roots of $f(t)$, it can be verified that $\alpha\beta^{-1}$ and $\alpha^{-1}\beta$ are roots of $g(t)=t^{2}+(a^{-1}b^{2}+2)t+1$. Therefore, $a^{-1}b^{2}+2=\alpha\beta^{-1}+\alpha^{-1}\beta$. Thus, $a=b^{2}(\alpha\beta^{-1}+\alpha^{-1}\beta-2)$. By the theory of Galois rings, $\alpha\beta^{-1}$ can be expressed as $\alpha\beta^{-1}=xy$, where $x\in\Gamma_{e,2}$ and $y\in 1+(p)$. Thus, ${\rm ord}(\alpha\beta^{-1})=kp^{i}$, where $1\leq i\leq e-1$. For each $kp^{i}$, there are $\varphi(kp^{i})$ elements whose order is $kp^{i}$ and there are $\frac{\varphi(kp^{i})}{2}$ different $\alpha\beta^{-1}+\alpha^{-1}\beta-2$ 's. Thus, there are $\sum^{e-1}_{i=0}\frac{\varphi(kp^{i})}{2}$ choices of $\alpha\beta^{-1}+\alpha^{-1}\beta-2 $.

As a result of ${\rm ord}(\omega_{1}\omega^{-1}_{2})>2$, we have $\alpha\beta^{-1}+\alpha^{-1}\beta-2$ is a unit. The number of choices of $b$ is $(p-1)p^{e-1}$. Once $b$ and $\alpha\beta^{-1}+\alpha^{-1}\beta-2$ are chosen, $a$ is uniquely determined. Hence, for each $k$, there are $(k-1)(p-1)p^{2e-2}\sum^{e-1}_{i=0}\frac{\varphi(kp^{i})}{2}$ IPRNGs of period $k$. The proof is completed.
\end{proof}
\begin{proposition}
Suppose $f(t)$ is irreducible in ${\rm Z}_{p^{e}}[t]$ and $p\nmid \alpha-\beta$. If $(\pi-\omega_{1})(\pi-\omega_{2})^{-1}\notin\Omega$, then $L(x_{0};a,b)$ traverses the set $\{k=k_{1}k_{2}:2<k_{1}<p+1,k_{1}\mid p+1,k_{2}\mid p^{e-1}\}$. For each $k$, there are $(p-(k_{1}-1))(p-1)\frac{\varphi(k)}{2}p^{2e-2}$ IPRNGs of period $k$.
\end{proposition}
\begin{proof}
Period analysis.

By the proof of Proposition 7, we can get that ${\rm ord}(\omega_{1}\omega^{-1}_{2})>2$.

Since $(\pi-\omega_{1})(\pi-\omega_{2})^{-1}\notin\Omega$, it is valid that $S(x_{0};a,b)$ does not contain any element in $(p)$. Thus, $L(x_{0};a,b)={\rm ord}(\alpha\beta^{-1})$. In this case, $(\pi-\omega_{1})(\pi-\omega_{2})^{-1}\notin \Omega$. Thus, ${\rm ord}(\omega_{1}\omega^{-1}_{2})\neq p+1$. By the proof of Proposition 7, we have ${\rm ord}(\omega_{1}\omega^{-1}_{2})\mid p+1$. By the theory of Galois rings, $\alpha\beta^{-1}$ can be expressed as $\alpha\beta^{-1}=xy$, where $x\in\Gamma_{e,2}$ and $y\in 1+(p)$. Thus, ${\rm ord}(\alpha\beta^{-1})$ traverses the set $\{k=k_{1}k_{2}:2<k_{1}<p+1,k_{1}\mid p+1,k_{2}\mid p^{e-1}\}$, so does $L(x_{0};a,b)$.

Counting.

For $L(x_{0};a,b)=k$, there are $p-(k_{1}-1)$ $\pi$'s such that $(\pi-\omega_{1})(\pi-\omega_{2})^{-1}\notin\Omega$. There are $p$ choices of $h_{i}$ for all $i=1,2,\ldots,e-1$. Thus, there are $(p-(k_{1}-1))p^{e-1}$ choices of $x_{0}$.

The rest of the counting process is the same as which in Proposition 7, thus, we omit it. There are $(p-(k_{1}-1))(p-1)\frac{\varphi(k)}{2}p^{2e-2}$ IPRNGs of period $k$. The proof is completed.
\end{proof}
\subsection{$\alpha-\beta$ Is a Zero or a Zero Divisor}
Denote ${\rm R}={\rm Z}_{p^{e}}[t]/(f(t))$. Let $\psi$ be the nature homomorphism between ${\rm R}$ and ${\rm R}/p{\rm R}$. If $p\mid \alpha-\beta$, then it holds that $\psi(\alpha-\beta)=\psi(\alpha)-\psi(\beta)=0$. By the analysis in \cite{c3}, we can get that ${\rm R}/p{\rm R}$ is isomorphic with ${\rm GF}(p^{2})$. Thus, $\psi(\alpha)=\psi(\beta)=\omega+p{\rm R}$, where $\omega\in {\rm Z}_{p}$. Since $a=-\alpha\beta$ and $b=\alpha+\beta$, it holds that $f(t)=t^{2}-2\omega t+\omega^{2}$ in ${\rm Z}_{p}$, which means that $f(t)$ is not a basic irreducible polynomial in ${\rm Z}_{p^{e}}$. Thus, ${\rm R}$ is not a Galois ring.

Denote $x_{0}=\pi+\sum^{e-1}_{i=1}h_{i}p^{i}$ where $\pi,h_{i}\in {\rm Z}_{p}$ for all $i=1,2,\ldots,e-1$. Then, it follows from recurrence relation (5) that
$$
y_{n+2}\equiv 2\omega y_{n+1}-\omega^{2}y_{n} ({\rm mod}p).
$$
Let $x'_{n}=x_{n} {\rm mod}p$ and $y'_{n}=y_{n} {\rm mod}p$ for all $n=0,1,\ldots$. Then, we obtain
\begin{eqnarray}
y'_{n+2}= 2\omega y'_{n+1}-\omega^{2}y'_{n}.
\end{eqnarray}

Similar to (6), we have the general term of (10)
\begin{eqnarray}
y'_{n}=\omega^{n}(n(\omega^{-1}\pi-1)+1).
\end{eqnarray}
Thus, if $\pi-\omega$ is a unit, then $y'_{n}$ must contain $0$, which means that $S(x_{0};a,b)$ must contain some elements in $(p)$; Otherwise, $y'_{n}$ dose not contain $0$, which means that $S(x_{0};a,b)$ does not contain any element in $(p)$.

Now, we are ready to present our results on the period distribution of IPRNGs for this case.
\begin{proposition}
Suppose $p\mid\alpha-\beta$. If $\pi-\omega\neq0$, then $L(x_{0};a,b)=p-1$. There are $(p-1)^{2}p^{3e-3}$ IPRNGs of period $p-1$.
\end{proposition}
\begin{proof}
Period analysis.

Since $\pi-\omega\neq0$, it is valid that $S(x_{0};a,b)$ contains some elements in $(p)$. Thus, $L(x_{0};a,b)=L(b;a,b)$. Then, we consider the case that $x_{0}=b$, which means that $\pi=2\omega$. By (11), we can get that $y'_{n}=(n+1)\omega^{n}$. Thus, $n'=p-1$ is the smallest integer such that $y'_{n'}=0$. It follows from Lemma 4 that $x'_{n'-1}=0$, which means that $x_{n'-1}\in(p)$. Thus, $x_{n'}=b$, which means that $L(b;a,b)=p-1$, so does $L(x_{0};a,b)$.

Counting.

For $L(x_{0};a,b)=p-1$, since $a,b\in{\rm Z}_{p^{e}}^{\times}$, it must hold that $\omega\in{\rm Z}^{\times}_{p}$. Thus, there are $p-1$ choices of $\omega$. Once $\omega$ is chosen, there are $p^{e-1}$ choices of $a,b$, respectively. Since $S(x_{0};a,b)$ contains some elements in $(p)$, it is valid that $\pi-\omega$ is a unit. there are $p-1$ choices of $\pi$, thus, there are $(p-1)p^{e-1}$ choices of $x_{0}$. Hence, there are $(p-1)^{2}p^{3e-3}$ IPRNGs of period $p-1$. The proof is completed.
\end{proof}
\begin{proposition}
Suppose $p\mid\alpha-\beta$. If either $x_{0}-\alpha$ or $x_{0}-\beta$ is a zero, then $L(x_{0};a,b)=1$. There are $(p-1)p^{2e-2}$ IPRNGs of period $1$. If both $x_{0}-\alpha$ and $x_{0}-\beta$ are zero divisors, then $L(x_{0};a,b)$ traverses set $\{p^{e-k}:1\leq k\leq e-1\}$. For each $k$, there are $(p-1)^{2}p^{3e-k-3}$ IPRNGs of period $p^{e-k}$.
\end{proposition}
\begin{proof}
Period analysis.

If $\pi-\omega=0$, then $S(x_{0};a,b)$ does not contain any element in $(p)$. Thus, $x_{n}=x_{0}$ if and only if
\begin{eqnarray}
(\alpha^{n-1}+\alpha^{n-2}\beta+\cdots+\beta^{n-1})(x_{0}-\alpha)(x_{0}-\beta)=0.
\end{eqnarray}

Since $\psi(\alpha)=\psi(\beta)=\omega+p{\rm R}$, we denote $\alpha=\omega+px$ and $\beta=\omega+py$, where $x,y\in{\rm R}$. Thus, by simple calculation, we can get hat
\begin{eqnarray}
(\alpha^{n-1}+\alpha^{n-2}\beta+\cdots+\beta^{n-1})=n\omega+npz,
\end{eqnarray}
where $z$ is an element in ${\rm R}$.

If either $x_{0}-\alpha$ or $x_{0}-\beta$ is a zero, which means that $(x_{0}-\alpha)(x_{0}-\beta)=0$, then the smallest $n$ such that (12) holds is $1$. Thus, $L(x_{0};a,b)=1$.

If both $x_{0}-\alpha$ and $x_{0}-\beta$ are zero divisors, we have $p\mid(x_{0}-\alpha)(x_{0}-\beta)$. Thus, we suppose $1\leq k\leq e-1$ is the largest integer such that $p^{k}\mid (x_{0}-\alpha)(x_{0}-\beta)$. Then $n=p^{e-k}$ is the smallest integer such that $(13)$ holds. Thus, $L(x_{0};a,b)=p^{e-k}$.

Counting.

For $L(x_{0};a,b)=1$, we have either $x_{0}-\alpha$ or $x_{0}-\beta$ is a zero. Since $\alpha,\beta\in {\rm Z}^{\times}_{p^{e}}$, it is valid that there are $(p-1)p^{e-1}$ choices of $\alpha$. Once $\alpha$ is chosen, if $\beta=\alpha$, then $\beta$ and $x_{0}$ are uniquely determined by a chosen $\alpha$. Therefore, there are $(p-1)p^{e-1}$ IPRNGs of period $1$ for this case. If $\alpha-\beta$ is a zero divisor, then there are $p^{e-1}-1$ choices of $\beta$. Thus, there are $\frac{(p-1)p^{e-1}(p^{e-1}-1)}{2}$ pairs of $a,b$. Once $\alpha,\beta$ are chosen, there are $2$ choices of $x_{0}$. Thus, there are $(p-1)p^{e-1}(p^{e-1}-1)$ IPRNGs of period $1$ for this case. Hence, there are $(p-1)p^{2e-2}$ IPRNGs of period $1$.

For $L(x_{0};a,b)=p^{e-k}$, we have $p^{k}\mid (x_{0}-\alpha)(x_{0}-\beta)$. Thus, there exists $\alpha',\beta'\in {\rm Z}^{\times}_{p^{e}}$ such that $(x_{0}-\alpha)(x_{0}-\beta)=(x_{0}-\alpha')(x_{0}-\beta')+\sum^{e-1}_{i=k}c_{i}p^{i}$, where $(x_{0}-\alpha')(x_{0}-\beta')=0$,  $c_{k}\in{\rm Z}^{\times}_{p}$ and $c_{i}\in{\rm Z}_{p}$ for all $i=k+1,k+2,\ldots,e-1$. By the counting process of $L(x_{0};a,b)=1$, we have there are $(p-1)p^{2e-2}$ $(x_{0}-\alpha')(x_{0}-\beta')$'s with $(x_{0}-\alpha')(x_{0}-\beta')=0$. Once $x_{0},\alpha,\beta$ are chosen, there are $p-1$ choices of $c_{k}$, $p$ choices of $c_{i}$ for all $i=k+1,k+2,\ldots,e-1$. Thus, there are $(p-1)p^{e-k-1}$ choices of $\sum^{e-1}_{i=k}c_{i}p^{i}$'s. Hence, there are $(p-1)^{2}p^{3e-k-3}$ IPRNGs of period $p^{e-k}$. The proof is completed.
\end{proof}
\begin{theorem}
For IPRNGs with $a\in{\rm Z}^{\times}_{p^{e}}$ and $b\in{\rm Z}^{\times}_{p^{e}}$, the possible periods and the number of each special period are given in Table IV.
\end{theorem}
\begin{table}[!t]
\renewcommand{\arraystretch}{1.3}
\caption{Period distribution of IPRNGs with $a\in{\rm Z}^{\times}_{p^{e}}$ and $b\in{\rm Z}^{\times}_{p^{e}}$ in ${\rm Z}_{p^{e}}$}
\label{table_example}
\centering
\begin{tabular}{|c|c|}
\hline
\bfseries Periods & \bfseries Number of IPRNGs\\
\hline
 \tabincell{c}{$1$} & $(p-2)(p-1)p^{2e-2}$\\
\hline
\tabincell{c}{$p-1$}& \tabincell{c}{$(p-1)^{2}p^{3e-3}$}\\
\hline
\tabincell{c}{$\{k-1:k>2,k\mid p-1\}$} & \tabincell{c}{$(k-1)(p-1)p^{2e-2}\sum^{e-1}_{i=0}\frac{\varphi(kp^{i})}{2}$}\\
\hline
\tabincell{c}{$\{k-1:k>2,k\mid p+1\}$}& \tabincell{c}{$(k-1)(p-1)p^{2e-2}\sum^{e-1}_{i=0}\frac{\varphi(kp^{i})}{2}$}\\
\hline
\tabincell{c}{$\{k=k_{1}k_{2}:2<k_{1}<p-1,k_{1}\mid p-1,$\\$k_{2}\mid p^{e-1}\}$} & \tabincell{c}{$(p-(k_{1}-1))(p-1)\frac{\varphi(k)}{2}p^{2e-2}$}\\
\hline
\tabincell{c}{$\{k=k_{1}k_{2}:k_{1}>2,k_{1}\mid p-1,$\\$k_{2}\mid p^{e-k_{3}-1},1\leq k_{3}\leq e-1\}$} & \tabincell{c}{$\varphi(k)(p-1)^{2}p^{2e-2}$}\\
\hline
\tabincell{c}{$\{k=k_{1}k_{2}:2<k_{1}<p+1,k_{1}\mid p+1,$\\$k_{2}\mid p^{e-1}\}$} & \tabincell{c}{$(p-(k_{1}-1))(p-1)\frac{\varphi(k)}{2}p^{2e-2}$}\\
\hline
\tabincell{c}{$\{p^{e-k}:1\leq k\leq e-1\}$}& \tabincell{c}{$(p-1)^{2}p^{3e-k-3}$}\\
\hline
\end{tabular}
\end{table}
\begin{remark}
It should be mentioned that $p>3$ is an important condition in Theorem 3, because of some periods require $k>2,k\mid p-1$, which implies that $p>3$.
\end{remark}
\begin{example}
The following example is given to compare experimental and the theoretical results. A computer program has been written to exhaust all possible IPRNGs with $a\in{\rm Z}^{\times}_{5^{3}}$ and $b\in{\rm Z}^{\times}_{5^{3}}$ in ${\rm Z}_{5^{3}}$ to find the period by brute force, the results are shown in Fig. 3.

Table V lists the complete result we have obtained. It provides full information on the period distribution of the IPRNGs. As it is shown in Fig. 3 and Table V, the theoretical and experimental results fit well. The maximal period is $75$ while the minimal period is $1$. The analysis process also indicates how to choose the parameters and the initial values such that the IPRNGs fit specific periods.
\end{example}
\begin{table}[!t]
\renewcommand{\arraystretch}{1.3}
\caption{Period distribution of IPRNGs with $a\in{\rm Z}^{\times}_{5^{3}}$ and $b\in{\rm Z}^{\times}_{5^{3}}$ in ${\rm Z}_{5^{3}}$}
\label{table_example}
\centering
\begin{tabular}{|c|c|c|c|c|}
\hline
Periods &1 &2 &3&4\\
\hline
Number of IPRNGs &7500&125000&195000&290000\\
\hline
Periods &5&10&20&25\\
\hline
Number of IPRNGs &322500 &30000&80000&50000\\
\hline
Periods &75& & & \\
\hline
Number of IPRNGs &150000 & & &\\
\hline
\end{tabular}
\end{table}
\section{Conclusion}
The period distribution of the IPRNGs over $({\rm Z}_{p^{e}},+,\times)$ for prime $p>3$ and integer $e\geq 2$ has been analyzed. Full information on the period distribution of IPRNGs is obtained by some analytical approaches. The analysis process also indicates how to choose the parameters and the initial values such that the IPRNGs fit specific periods.

\ifCLASSOPTIONcaptionsoff
\newpage
\fi

\end{document}